\documentclass[12pt]{article}

\usepackage[colorlinks,citecolor=blue,urlcolor=blue]{hyperref}

\usepackage{natbib}
\usepackage[english]{babel}
\usepackage{amsmath}
\usepackage{amssymb}
\usepackage{amsthm}
\usepackage{mathtools}
\usepackage{epsfig}
\usepackage{bm}
\usepackage{tikz}
\usepackage{subcaption}
\usepackage{soul}

\theoremstyle{plain}
\newtheorem{theorem}{Theorem}
\newtheorem{lemma}{Lemma}   
\newtheorem{corollary}{Corollary}

\usepackage{xcolor}
\usepackage{tabularx}
\usepackage{multirow}
\usepackage{booktabs}
\usepackage{rotating}
\usepackage{caption}
\usepackage{floatrow}
\usepackage{graphicx} 
\usepackage{pgfplots}
\usepackage{authblk}
\usepackage{etoolbox}

\usepackage{soul}

\pgfplotsset{compat=1.5}

\theoremstyle{plain}

\title{Improved Risk Ratio Approximation by Complementary Log-Log Models: A Comparison with Logistic Models
\thanks{Yuji Tsubota is the corresponding author of this article. Email: \href{u922531f@ecs.osaka-u.ac.jp}{u922531f@ecs.osaka-u.ac.jp} }
}
\date{}

\author[1]{Yuji Tsubota} \affil[1]{Graduate School of Human Sciences, Osaka University, Japan} \author[2]{Kenji Beppu} \affil[2]{Graduate School of Engineering Science, Osaka University, Japan}

\begin{document}
\maketitle
\section*{Abstract}
Odds ratios obtained from logistic models fail to approximate risk ratios with common outcomes, leading to potential misinterpretations about exposure effects by practitioners.
This article investigates the complementary log-log models as a practical alternative to produce risk ratio approximation.
We demonstrate that the corresponding effect measure of complementary log-log models, called the complementary log ratio in this article, consistently provides a closer approximation to risk ratios than odds ratios.
To compare the approximation accuracy, we adopt the one-parameter Aranda-Ordaz family of link functions, which includes both the logit and complementary log-log link functions as special cases.
Within this unified framework, we implement a theoretical comparison of approximation accuracy between the complementary log ratio and the odds ratio, showing that the former always produces smaller approximation bias.
Simulation studies further reinforce our theoretical findings.
Given that the complementary log-log model is easily implemented in standard statistical software such as R and SAS, we encourage more frequent use of this model as a simple and effective alternative to logistic models when the goal is to approximate risk ratios more accurately.

\vspace{0.1cm}

\noindent{\bf keywords}: Odds Ratios, Complementary log-log model, Risk Ratio Approximation, Aranda-Ordaz family of link functions

\section{Introduction}
When the outcome of interest is dichotomous, odds ratios are frequently reported as a measure of exposure effects in cohort studies and randomized controlled trials \citep{zhang1998s, knol2011potential, vanderweele2020optimal}.
This widespread use of odds ratios comes mainly from the popularity of logistic regression analyses \citep{robbins2002s, penman2009complementary}.

However, the interpretation of odds ratios as a measure of exposure effect is not straightforward and is often misleading \citep{zhang1998s, robbins2002s, penman2009complementary}.
With rare outcomes, odds ratios closely approximate risk ratios, enabling a straightforward interpretation of exposure effects as risk ratios.
On the other hand, such an interpretation is no longer valid when the outcome is common \citep{vanderweele2020optimal}.

A substantial body of research has pointed out the problem of misinterpreting odds ratios as risk ratios in practice with common outcomes \citep[e.g.,][]{zhang1998s, altman1998odds, robbins2002s, knol2011potential, knol2012overestimation, vanderweele2020optimal}.
\cite{zhang1998s} visually illustrated that, when true risk ratios are greater (or less) than \(1\), the corresponding odds ratios always overestimate (or underestimate) the values of risk ratios.
Resulting deviations of odds ratios from risk ratios become significant when the outcome prevalence is greater than \(10\%\) \citep{zhang1998s, knol2012overestimation, hosmer2013applied}. 
Therefore, existing literature tends to use \(10\%\) as a cutoff outcome prevalence where odds ratios can be safely interpreted as risk ratios \citep{robbins2002s, hosmer2013applied}.

Several studies have suggested alternative approaches for estimating covariate-adjusted exposure effects on binary outcomes \citep{zhang1998s, zou2004modified, penman2009complementary, richardson2017modeling}.
However, each alternative approach has its own limitations and drawbacks, such as convergence issues, inability to accommodate interactions, or theoretical complexity. 
Moreover, among such studies, there has been limited investigation of other standard binary generalized linear models \citep[GLMs;][]{nelder1972generalized} except for log-binomial models that directly estimate risk ratios \citep{penman2009complementary}.

In this article, we investigate the potential utility of a less-utilized binary GLM, complementary log-log models \citep{fisher1922mathematical}, in approximating risk ratios. 
We compare odds ratios with the corresponding estimand of 
complementary log-log models, and show that the latter estimand is always a better approximation of risk ratios than odds ratios.
In our mathematical proof, we introduce the one-parameter family of link functions proposed by \cite{aranda1981two} that contains logit and complementary log-log link functions as special cases.
Following the framework by \cite{aranda1981two}, we can express odds ratios and the corresponding estimand of complementary log-log models in a unified way, thereby enhancing the clarity and brevity of our mathematical derivation.

The present article does not argue against the direct estimation of risk ratios. In fact, it is generally preferable when such estimation is stable and permits valid inference.  However, it can involve complex modeling or computational challenges \citep{williamson2013log}.
To address these difficulties, our goal is to suggest a more accessible and practical alternative based on standard GLMs.

The remainder of this paper is organized as follows. Section \ref{setup} introduces the Aranda-Ordaz transformation family, preparing essential theoretical groundwork for the mathematical analyses in later sections.
In Section \ref{Sec::theoretical_comparison}, we provide a theoretical comparison of approximation accuracy between odds ratios and the corresponding effect measures of complementary log-log models within this framework. Section \ref{discussion} presents concluding remarks. 
All technical proofs are provided in the Supplementary Materials.

\section{Setup} \label{setup}
\subsection{Binary Effect Measures in GLM literature}
Let \(A\) be a binary exposure and \(Y\) be a binary outcome of interest.
Define \(p_1 = P(Y = 1 | A = 1)\) and \(p_0 = P(Y = 1|A = 0)\) denoting the probability of having the outcome when exposed and unexposed, respectively.
Given specific values of \(p_1\) and \(p_0\), the risk ratio \(\mathrm{RR}\) is defined by \(\mathrm{RR} = p_1/p_0\).
Additionally, we define the odds ratio \(\mathrm{OR}\) and complementary log ratio \(\mathrm{CLR}\) for \(p_1\) and \(p_0\) respectively as 
\begin{align}
\mathrm{OR} = \frac{p_1/(1-p_1)}{p_0/(1-p_0)},\quad \mathrm{CLR} = \frac{\log(1-p_1)}{\log(1-p_0)}. \label{ORHR}  
\end{align}
Analogous to the odds ratio, we term the corresponding effect measure from complementary log-log models the ``complementary log ratio".
The definition of complementary log ratios in \eqref{ORHR} expresses the effect of exposure as a power function \citep{agresti2010analysis, agresti2012categorical}. 
Additionally, the same expression as complementary log ratios has appeared as an alternative expression of hazard ratios in some literature \citep{agresti2010analysis, vanderweele2020optimal}.

When the outcome is rare, for example, both \(p_1\) and \(p_0\) are less than or equal to \(0.1\), it is not difficult to see \(\mathrm{OR} \approx \mathrm{RR}\) \citep{hosmer2013applied}.
On the other hand, for \(\mathrm{CLR}\), Maclaurin expansions of the numerator and the denominator give the approximation relationship \(\mathrm{CLR} \approx \mathrm{RR}\) with relatively rare outcomes \citep{vanderweele2020optimal}.
The above approximations hold for small outcome prevalence, but the values of \(\mathrm{OR}\) and \(\mathrm{CLR}\) significantly diverge from that of \(\mathrm{RR}\) when the outcome is common \citep{vanderweele2020optimal}.

\subsection{The family of Aranda-Ordaz transformations}

Because both the logistic and complementary log-log models belong to the Aranda-Ordaz family of transformations \citep{aranda1981two}, we introduce this parametric family as a unifying framework. 
This unified framework enables systematic comparisons between logistic and complementary log-log models, facilitating theoretical analyses of how closely each model can approximate risk ratios through varying a transformation parameter.

\cite{aranda1981two} considers the following family of transformations, which includes both logistic and complementary log-log models:
\begin{align*}
  W_{\lambda}(\theta) := 
  \begin{cases}
\{(1 - \theta)^{-\lambda} - 1\} /{\lambda}, & \text{when } 0<\lambda\leq 1, \\
-\log(1-\theta), & \text{when } \lambda = 0.
\end{cases}
\end{align*}
where $0<\theta<1$ denotes the probability of success and $0\leq \lambda\leq 1$ is the transformation parameter.
Note that the case of $\lambda=0$ is naturally defined in a mathematical sense.
More specifically, $\lim_{\lambda  \downarrow 0} W_{\lambda}(\theta)=W_{0}(\theta)$ holds.
This family $ W_{\lambda}(\theta)$ satisfies
\begin{align*}
  W_1(\theta) = \frac{\theta}{1 - \theta}, \quad W_{0}(\theta) = -\log(1 - \theta).
\end{align*}
When $\log \mathrm{W}_\lambda(\theta)$ is used as a link function for GLMs, the resulting models includes both the logistic model ($\lambda = 1$) and complementary log-log model ($\lambda = 0$) as special cases \citep{aranda1981two}.

Given that $\log W_{\lambda}(\theta)$ forms a parametric family of link functions containing the logit and complementary log-log links, it is natural to investigate how measures of association behave under the transformation  $\mathrm{W}_\lambda(\theta)$. 
In particular, we define a generalized ratio on the $W$-scale that unifies the odds ratio and complementary log ratio within a common framework.
Specifically, for $0 < p_0, p_1 < 1$, we define
\begin{align*}
  \mathrm{WR}(\lambda) := \frac{W_{\lambda}(p_1)}{W_{\lambda}(p_0)} .
\end{align*}
This transformation-based ratio $\mathrm{WR}(\lambda)$ generalizes the two measures of our interest: it coincides with the odds ratio when $\lambda = 1$ and the complementary log ratio when $\lambda  = 0$:
\begin{align*}
  \mathrm{WR}(1) = \frac{p_1/(1 - p_1)}{p_0/(1 - p_0)} = \mathrm{OR},\ \mathrm{WR}(0) &= \frac{\log(1 - p_1)}{\log(1 - p_0)} = \mathrm{CLR}.
\end{align*}
In practice, $\mathrm{WR}(\lambda)$ can be estimated by using \(\log \mathrm{W}_{\lambda}(\theta)\) as a link function for a GLM \citep{aranda1981two}.

Figure \ref{fig:RR_approx} shows the values of \(\mathrm{CLR}\), \(\mathrm{WR}(0.5)\) and \(\mathrm{OR}\) when fixing \(\mathrm{RR}\) to \(1.25\) or \(0.5\).
Note that \(\mathrm{WR}(\lambda)\) coincides with \(\mathrm{CLR}\) when \(\lambda = 0\), and with \(\mathrm{OR}\) when \(\lambda = 1\).
In both graphs, monotonic behaviors of \(\mathrm{WR}(\lambda)\) according to the increase in outcome prevalence are common for all \(\lambda\).
Additionally, we observe that when \(\mathrm{RR} > 1\) (or \(\mathrm{RR} < 1\)), \(\mathrm{WR}(\lambda)\) always overestimates (or underestimates) \(\mathrm{RR}\).
Moreover, the extent of such overestimation (or underestimation) decreases when the transformation parameter \(\lambda\) becomes small.
In the next section, we offer a mathematical justification for this monotonic behavior of \(\mathrm{WR}(\lambda)\) observed in Figure \ref{fig:RR_approx}.
\begin{figure}[bt]

  \begin{subfigure}[b]{0.49\textwidth} \includegraphics[width=\textwidth]{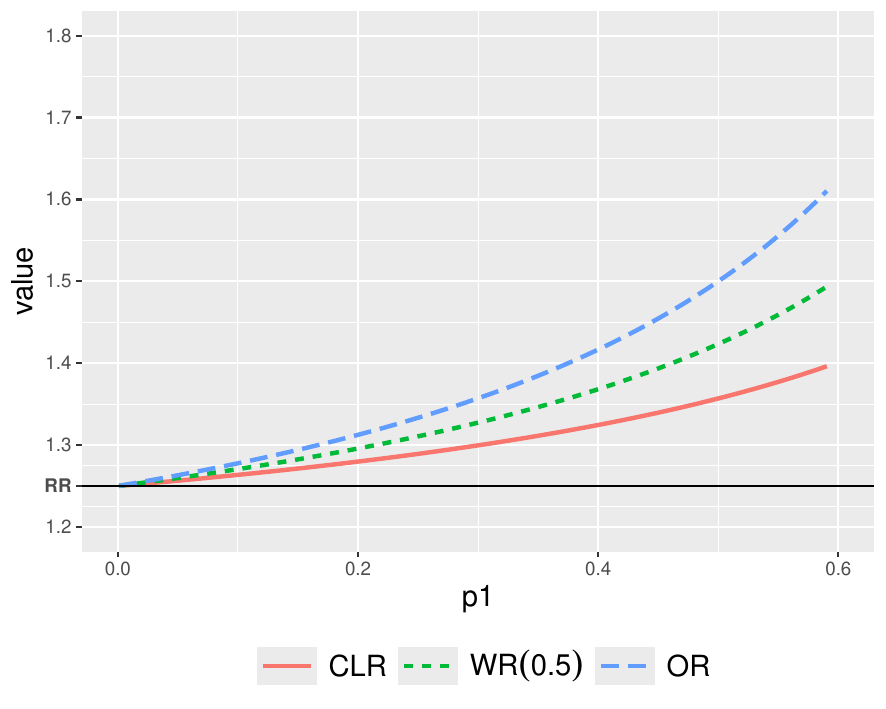}
  \caption{\(\mathrm{RR} = 1.25\)}
  \end{subfigure}
  \hfill
  \begin{subfigure}[b]{0.49\textwidth} \includegraphics[width=\textwidth]{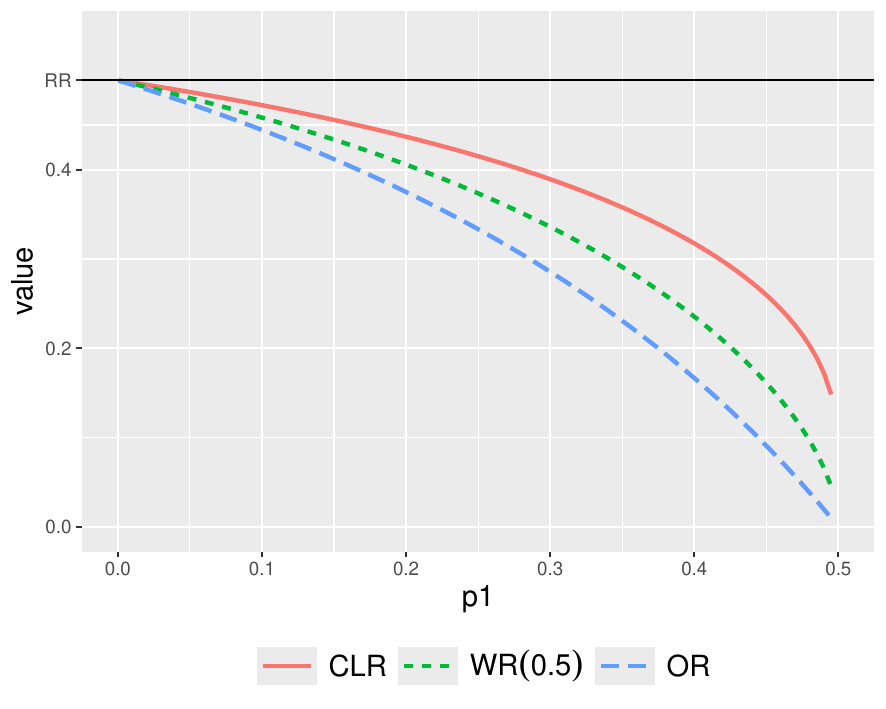}
  \caption{\(\mathrm{RR} = 0.5\)}
  \end{subfigure}
  \caption{Risk Ratio Approximation by Aranda-Ordaz Transformation Family}
  \label{fig:RR_approx}
\end{figure}

\section{Theoretical Comparison of Risk Ratio Approximations within Aranda-Ordaz transformation family}\label{Sec::theoretical_comparison}
Since both the logistic and complementary log-log models are special cases of the Aranda-Ordaz transformation family, we develop a unified theoretical framework for comparison using the transformation parameter $\lambda$.
This enables us to rigorously analyze how well each model approximates the risk ratio.
First, we consider the relative discrepancy between the risk ratio $\mathrm{RR}$ and the transformation-based ratio $\mathrm{WR}(\lambda)$.
Specifically, we define
\begin{align*} B(\lambda) := \max \left\{ \frac{\mathrm{RR}}{\mathrm{WR}(\lambda)}, \ \frac{\mathrm{WR}(\lambda)}{\mathrm{RR}} \right\} = \max \left\{ \frac{p_1/p_0}{W_{\lambda}(p_1)/W_{\lambda}(p_0)}, \ \frac{W_{\lambda}(p_1)/W_{\lambda}(p_0)}{p_1/p_0} \right\}.
\end{align*} 

By construction, $B(\lambda) \geq 1$ for all $\lambda \in [0,1]$, and the closer $B(\lambda)$ is to 1, the better $\mathrm{WR}(\lambda)$ approximates the risk ratios.
The lemma below describes the connection between $B(\lambda)$ and risk ratios.

\begin{lemma} \label{lemma_cha_B}
    Under $\mathrm{RR}>1$, which is equivalent to $p_0<p_1$, $\mathrm{RR}<\mathrm{WR}(\lambda)$ holds for all $\lambda\in[0,1]$, thus, we obtain $B(\lambda)=\mathrm{WR}(\lambda)/\mathrm{RR}$.
    Also, under $\mathrm{RR}<1$, which is equivalent to $p_1<p_0$, $\mathrm{WR}(\lambda)<\mathrm{RR}$ holds for all $\lambda\in[0,1]$, thus, we obtain $B(\lambda)=\mathrm{RR}/\mathrm{WR}(\lambda)$.
\end{lemma}

Lemma \ref{lemma_cha_B} indicates that, under  $\mathrm{RR}>1$ (or $\mathrm{RR}<1$), $\mathrm{WR}(\lambda)$ always overestimates (or underestimates) the risk ratios. 
Since $\mathrm{WR}(\lambda)$ with $\lambda=1$ corresponds to the odds ratios, this result generalizes the well-known fact that the odds ratios always overestimate (or underestimate) the risk ratios under $\mathrm{RR}>1$ (or $\mathrm{RR}<1$) \citep{zhang1998s}.

Using lemma \ref{lemma_cha_B}, the following theorem characterizes the monotonic behavior of $B(\lambda)$:

\begin{theorem} \label{main_theorem}
Fix any $0 < p_0 \neq p_1 < 1$. Then the function $B(\lambda)$ is strictly increasing over the interval $0 \leq  \lambda \leq 1$. Furthermore, if $p_0 = p_1$, then $B(\lambda) = 1$ for all $\lambda$, i.e., $B(\lambda)$ is constant. 
\end{theorem}

The monotonicity of $B(\lambda)$ established in Theorem \ref{main_theorem} implies that, for fixed values of $p_0$ and $p_1$, models with smaller values of $\lambda$ produce better approximations to the risk ratios.
The following corollary compares the approximation accuracy of the risk ratios under the logistic and complementary log-log models, both of which are included in the Aranda-Ordaz family of link functions. It follows directly from Theorem \ref{main_theorem} by evaluating the result at $\lambda = 0$ and $\lambda = 1$.

\begin{corollary} \label{corollary_CLR}
Fix any $0 < p_0 \neq p_1 < 1$. Then the following inequality holds:
\begin{align*}
\max \left\{ \frac{p_1/p_0}{\mathrm{{CLR}}}, \ \frac{\mathrm{{CLR}}}{p_1/p_0} \right\} 
< 
\max \left\{ \frac{p_1/p_0}{\mathrm{OR}}, \ \frac{\mathrm{OR}}{p_1/p_0} \right\},
\end{align*}
where $\mathrm{{CLR}}$ and $\mathrm{OR}$ are defined as in \eqref{ORHR}.
\end{corollary}
Corollary \ref{corollary_CLR} formally establishes that, for any values of $p_0$ and $p_1$ in the unit interval, the CLR consistently provides a more accurate approximation to the RR than the OR, in terms of maximum relative discrepancy.
This result offers a theoretical basis for preferring CLR over OR when the objective is to approximate RR.
Moreover, Theorem \ref{main_theorem} implies that the complementary log-log model, corresponding to $\lambda = 0$, achieves the smallest approximation error $B(\lambda)$ within the Aranda-Ordaz transformation family.

\section{Discussion} \label{discussion}

This study revisited the issue concerning risk ratio approximation in binary outcome analyses and highlighted the potential advantages of using complementary log-log models within the Aranda-Ordaz transformation framework.
We theoretically compared odds ratios from logistic models with complementary log ratios from complementary log-log models.
Our results in Section \ref{Sec::theoretical_comparison} established that the complementary log ratios consistently yield a closer approximation to risk ratios than odds ratios.

In contrast to various methods that estimate risk ratios directly, the complementary log-log model serves as a notably simple alternative. Importantly, it can be implemented using standard statistical software such as R or SAS.
For example, in R, users can obtain estimates from complementary log-log models by simply specifying \texttt{family = binomial(link = "cloglog")} in the \texttt{glm()} function.
The accessibility of complementary log-log models makes it particularly attractive in applied settings.
Since the complementary log-log model is also a standard generalized linear model \citep{agresti2012categorical}, researchers and practitioners already familiar with the logistic model can easily adopt this alternative without additional computational burden and necessity to learn new theoretical concepts. 
Our findings thus promote complementary log-log regression analyses as a practical substitute for logistic regression analyses when producing better risk ratio approximation is desired.

It should be emphasized that this research does not take a stance against the direct estimation of risk ratios.
On the contrary, when such approaches are computationally feasible and methodologically appropriate, direct estimation represents a natural and ideal strategy.
In light of the computational and theoretical burdens that may arise in directly estimating risk ratios, a method that provides more robust approximation while remaining simple to implement in practice may offer considerable benefits for applied researchers.

\clearpage
\bibliographystyle{chicago} 
\bibliography{ref}

%------------------ Figure ------------------------------

%------------------ APPENDICES ------------------------------
\newpage

\section*{eAppendix A.1: Technical proofs}

In this section, we provide the technical proofs of Lemma \ref{lemma_cha_B} and Theorem \ref{main_theorem}.

\begin{proof}[Proof of Lemma \ref{lemma_cha_B}]

Due to the continuity of $\mathrm{WR}(\lambda)$ at $\lambda = 0$, it is sufficient to prove the case of $0 < \lambda \leq 1$.
Without loss of generality, we consider the case where $0 < p_0 < p_1 < 1$.
By differentiating $W_{\lambda}(\theta)$ twice with respect to $\theta$, we get 
\begin{align*}
  \frac{\partial^2}{\partial \theta^2} W_{\lambda}(\theta) = \lambda (1 - \theta)^{-\lambda} > 0,
\end{align*}
which implies that $W_{\lambda}(\theta)$ is convex in $\theta$. Then, by the convexity of $W_{\lambda}$, it follows that for any $0 < p_0 < p_1$,
\begin{align*}
  \frac{W_{\lambda}(p_0) - W_{\lambda}(0)}{p_0 - 0} \leq \frac{W_{\lambda}(p_1) - W_{\lambda}(0)}{p_1 - 0}.
\end{align*}
Noting that $W_{\lambda}(0) = 0$ and simplifying, we obtain
\begin{align*}
  \frac{p_1}{p_0} \leq \frac{W_{\lambda}(p_1)}{W_{\lambda}(p_0)}.
\end{align*}
Therefore, when $0 < p_0 < p_1 < 1$, we have
\begin{align}
B(\lambda) := \max \left\{ \frac{p_1/p_0}{W_{\lambda}(p_1)/W_{\lambda}(p_0)}, \ \frac{W_{\lambda}(p_1)/W_{\lambda}(p_0)}{p_1/p_0} \right\} = \frac{W_{\lambda}(p_1)/W_{\lambda}(p_0)}{p_1/p_0} = \frac{\mathrm{WR}(\lambda)}{p_1/p_0}. \label{key_ineq}
\end{align}
\end{proof}

\begin{proof}[Proof of Theorem \ref{main_theorem}]

Similar to the proof of Lemma \ref{lemma_cha_B}, it is sufficient to prove the case of $0 < \lambda \leq 1$ and $0 < p_0 < p_1 < 1$.
The monotonicity of $B(\lambda)$ can be derived from that of $\mathrm{WR}(\lambda)$ from the equation \eqref{key_ineq}.
Hence, it suffices to prove that $\mathrm{WR}(\lambda)$ is monotonic.

We define $a := -\ln(1 - p_0),\  b := -\ln(1 - p_1),$ so that $0<a < b$.
By the properties of the exponential function, we have
$(1 - p_0)^{-\lambda} = e^{\lambda a}, \ (1 - p_1)^{-\lambda} = e^{\lambda b}$.
Thus, the expression for $\mathrm{WR}(\lambda)$ can be rewritten as
\[
\mathrm{WR}(\lambda) = \frac{e^{\lambda b} - 1}{e^{\lambda a} - 1}.
\]
Taking the logarithm of both sides, we obtain
\[
\ln \mathrm{WR}(\lambda) = \ln\bigl(e^{\lambda b} - 1\bigr) - \ln\bigl(e^{\lambda a} - 1\bigr).
\]
Differentiating with respect to $\lambda$ yields
\[
\frac{d}{d\lambda} \ln \mathrm{WR}(\lambda)
= \frac{b\,e^{\lambda b}}{e^{\lambda b} - 1} - \frac{a\,e^{\lambda a}}{e^{\lambda a} - 1}.
\]
Therefore, a sufficient condition for $\ln \mathrm{WR}(\lambda)$ to be strictly increasing is
\begin{align}
 \frac{b\,e^{\lambda b}}{e^{\lambda b} - 1} > \frac{a\,e^{\lambda a}}{e^{\lambda a} - 1}.  \label{sufficient_cond}    
\end{align}
Now, we define the function
\[
h(x) := \frac{x\,e^x}{e^x - 1}, \quad \text{for } x > 0,
\]
and its derivative is given by
\[
h'(x) = \frac{e^{2x} - e^x(x + 1)}{(e^x - 1)^2}
= \frac{e^x \bigl(e^x - (x + 1)\bigr)}{(e^x - 1)^2}.
\]
Since $e^x > x + 1$ for all $x > 0$, it follows that $e^x - (x + 1) > 0$, and hence $h'(x) > 0$. That is, $h(x)$ is strictly increasing for $x > 0$.
Since $\lambda a < \lambda b$, we get
\[
h(\lambda a) = \frac{\lambda a\,e^{\lambda a}}{e^{\lambda a} - 1}
< \frac{\lambda b\,e^{\lambda b}}{e^{\lambda b} - 1} = h(\lambda b).
\]
Therefore, the above equation is equivalent to the equation \eqref{sufficient_cond} which is the sufficient condition for the monotonicity of $\mathrm{WR}(\lambda)$.
By the equation \eqref{key_ineq}, we prove the monotonicity of $B(\lambda)$.

\end{proof}

\end{document}